\documentclass[letterpaper, 10 pt, conference]{ieeeconf}  

\IEEEoverridecommandlockouts                              
\usepackage{cite}
\usepackage{amsmath,amssymb}
\usepackage{algorithm,algorithmic}
\usepackage{hyperref}
\usepackage{textcomp}
\usepackage{optidef}
\usepackage{tikz}
\usetikzlibrary{calc}
\usetikzlibrary{positioning}
\usetikzlibrary{arrows.meta}

\newtheorem{prb}{Problem}
\newtheorem{prp}{Proposition}
\newtheorem{asm}{Assumption}
\newtheorem{dfn}{Definition}
\newtheorem{thm}{Theorem}
\newtheorem{lmm}{Lemma}
\newtheorem{rem}{Remark}

\title{\LARGE \bf
Data Poisoning Attacks on Informativity for Observability: Invariance-Based Synthesis 
}

\author{Iori Takaki, Ahmet Cetinkaya, and Hideaki Ishii
\thanks{This work was supported in part by JSPS under Grant-in-Aid for
Scientific Research Grant Nos. 24K00844 and 23K03913, and by JST
under ASPIRE Grant No. JPMJAP2402.}
\thanks{I. Takaki and H. Ishii are with the Department of Information Physics and Computing, The University of Tokyo, Tokyo 113-8656, Japan. (email:
        {\tt\small ioritakaki@g.ecc.u-tokyo.ac.jp, hideaki\_ishii@ipc.i.u-tokyo.ac.jp})}%
\thanks{A. Cetinkaya is with the Department of Functional Control Systems, Shibaura Institute of
Technology, Tokyo 135-8548, Japan. (email:
        {\tt\small ahmet@shibaura-it.ac.jp})}%
}

\begin{document}

\maketitle
\thispagestyle{empty}
\pagestyle{empty}

\begin{abstract}
This paper studies cyber attacks against informativity-based analysis in data-driven control. Focusing on strong observability, we consider an adversary who post-processes finite time-series data by an invertible linear transformation acting on the data matrices. We show that such transformations are capable of embedding malicious states into the invariant subspace explained by the transformed dataset. We provide a constructive attack method and derive feasibility conditions that characterize when such transformations exist. Moreover, we formulate an optimization problem to obtain the minimum-norm attack that quantifies the smallest data distortion required to destroy informativity. Numerical examples demonstrate that small and structured transformations can invalidate informativity certificates. 
\end{abstract}

\section{Introduction}
Data-driven control has rapidly evolved as a general framework that directly leverages measured data to analyze theoretical properties of a system, such as stability, observability, and controllability, and to synthesize controllers without committing to a single identified model \cite{DePersis}. Within this framework, data informativity provides a key theoretical foundation \cite{Informativity}. It characterizes whether the available data are sufficient to solve a control or analysis task, by reasoning uniformly over the model set consistent with the data. Many analysis results are expressed via rank-type conditions on data matrices, while synthesis problems can often be cast as robust design over the model set and then be reduced to linear matrix inequalities (LMIs) feasibility problems \cite{Informativity, Mishra}. Under independent random noise in the data, such rank conditions are often hard to break intentionally, so informativity-based certification may appear robust. Furthermore, applications to networked control systems can be found, including data informativity required for distributed control system design \cite{DCS}, cooperative control in multi-agent systems \cite{MAS}, and quadratic stabilization using quantized data \cite{Necsys25}. 

However, in cyber-physical environments, the presence of cyberattacks mimicking process/measurement noise is also conceivable. Previous research has discussed data manipulation aimed at system identification \cite{Russo} or closed-loop destabilization \cite{Sasahara}. In particular, when adversaries can execute data poisoning, sophisticated structured manipulations may inject undesirable, and potentially malicious dynamic behavior into models explained by the data. This not only invalidates rank-based analysis but may also destroy stabilization conditions required for LMI-based synthesis, rendering the design problem infeasible. Other related research includes data poisoning against data-driven parameter tuning \cite{Russo Proutiere}, identification of attack-free sensors using data \cite{Anand}, reconstruction of original data from falsified dataset \cite{Yan}, and the resilience of systems against data poisoning \cite{Shinohara}.

In the data-informativity framework, a related but distinct perspective comes from the \emph{data transformation} viewpoint \cite{TanakaKanekoKyb}. Motivated by the role of equivalent transformations in model-based control, the recent work \cite{TanakaKanekoKyb} formalizes how transforming measured time-series data induces a corresponding transformation of the compatible model set. While existing techniques are primarily developed as tools to simplify analysis or preserve specific properties of the systems, they also suggest an attack surface: an adversary may intentionally apply an \emph{invertible} transformation that preserves superficial data richness (e.g., rank/persistency of excitation) yet changes the induced model set in a harmful direction.

Motivated by this observation, we study \emph{malicious data transformations} against informativity-based analysis, with a focus on strong observability. Concretely, we consider an attacker who post-processes the pre-training dataset by an invertible block linear map acting separately on each of the data matrices, thereby avoiding trivial rank-reducing manipulations and enhancing stealthiness. The attacker's goal is to render the transformed data not informative for strong observability. This is carried out by embedding a nonzero state into the data-driven weakly unobservable subspace associated with the transformed dataset \cite{Eising}. Such an attack results in weakly unobservable models being introduced into the models explained by the data, making it impossible to analyze the properties of the true model from the transformed dataset. Furthermore, many data informativity-based design problems based on common Lyapunov functions (especially for observer design \cite{Mishra}) become impossible.

Our contributions are threefold. First, building on the invariance/subspace characterization of strong observability in affine model sets \cite{EisingECC, Eising}, we derive a constructive transformation-based attack that injects a designated weakly unobservable eigenpair into the post-attack model set and enlarges the maximum weakly unobservable coefficient space. As a result, the resulting dataset still appears informative (i.e., the transformation is nonsingular and preserves rank), but the informativity test against strong observability will fail. Second, we provide equivalent geometric and rank-type feasibility conditions under which such invertible transformations exist. These conditions explain when the above attack is possible purely from the structure of the data matrices. They also lead to simple, data-checkable tests that tell whether a given dataset is vulnerable to transformation-based tampering. Third, we develop a minimum-norm attack variant that quantifies the smallest data perturbation required to destroy informativity, offering a principled vulnerability metric for informativity-based certification \cite{Eising}. We cast this as an optimization problem and propose an efficient alternating algorithm to compute a near-minimal perturbation. The resulting value provides a clear \emph{distance-to-non-informativity}, which can be used as a quantitative robustness measure for data-driven certification.

The rest of the paper is organized as follows. In Section~II, we introduce the problem setting and formulate malicious data-transformation attacks. Section~III presents the main theoretical results, including feasibility conditions and a constructive design of invertible block-diagonal transformations. In Section~IV, we develop a minimum-norm attack formulation and provide an algorithm for computing stealthy transformations. Numerical examples are reported in Section~V to illustrate the effectiveness and interpretability of the proposed attacks. Finally, we conclude the paper in Section~VI. Due to space limitations, several technical results are stated without proofs.

\section{Problem formulation}
\subsection{Preliminaries}
First, we will present the problem setting and results from prior works. The framework of this study is based on the model representation using affine sets employed in \cite{Eising}. 

Consider the following linear time-invariant system:
\begin{equation}
\label{plant}
\begin{split}
x_{k+1}&=A_{\text{true}}x_k+Bu_k+Ew_k,\\
y_k&=Cx_k+Du_k+Fw_k,\\
\end{split}
\end{equation}
where $x_k\in\mathbb{R}^n$ is the state, $u_k\in\mathbb{R}^m$ is the input, $y_k\in\mathbb{R}^p$ is the output, and $w_k\in\mathbb{R}^l$ is the noise added during data collection.
Here, we assume that $A_{\text{true}}$ is unknown, but all other matrices $B$, $C$, $D$, $E$, and $F$ are known. Suppose that we have the following data matrices from this system:
\begin{equation}
\label{data}
\begin{split}
&X_-:=\begin{bmatrix} x_0 & x_1 & \cdots & x_{T-1}\end{bmatrix},\\
&X_+:=\begin{bmatrix} x_1 & x_2 & \cdots & x_T\end{bmatrix},\\
&U_-:=\begin{bmatrix} u_0 & u_1 & \cdots & u_{T-1}\end{bmatrix},\\
&Y_-:=\begin{bmatrix} y_0 & y_1 & \cdots & y_{T-1}\end{bmatrix},
\end{split}
\end{equation}
where $T$ is the data collection period. Let $\mathbb{D}:=\mathbb{R}^{n\times T}\times\mathbb{R}^{n\times T}\times\mathbb{R}^{m\times T}\times\mathbb{R}^{p\times T}$ denote the entire set of data, and let $\mathcal{D}:=\begin{bmatrix}
    X_-^\mathsf{T}&X_+^\mathsf{T} & U_-^\mathsf{T} & Y_-^\mathsf{T}
\end{bmatrix}^\mathsf{T}\in\mathbb D$ denote the matrix formed by stacking all data in (\ref{data}). We assume that the attacker can access these data. Additionally, although $w_k$ containing the process noise and the measurement noise is unknown, we similarly express it as the following matrix:
\begin{equation}
\label{noise}
W_-:=\begin{bmatrix} w_0 & w_1 & \cdots & w_{T-1}\end{bmatrix}.
\end{equation}
For further analysis, a model representation with noise removed is desirable. To this end, we adopt the following assumption based on \cite{Eising}.
\begin{asm}
\label{assumption A_dat}
There exist $M\in\mathbb{R}^{l\times n}$ and $N\in\mathbb{R}^{l\times p}$ such that
\begin{equation}
\mathrm{ker}\begin{bmatrix}M & N\end{bmatrix}=\mathrm{im}\begin{bmatrix}E\\F\end{bmatrix}.
\end{equation}
\end{asm}
\vspace{1mm}

Here, we define the matrices $P$, $Q$, and $R$ as
\begin{equation}
\label{affine parameters}
\begin{split}
&P(\mathcal{D}):=X_-,\;Q:=M,\;\\&R(\mathcal{D}):=\begin{bmatrix}
    M & N
\end{bmatrix}\begin{bmatrix}
    X_+-BU_-\\Y_--CX_--DU_-
\end{bmatrix}.
\end{split}
\end{equation}
Using these matrices, we express the set of models for the matrix $A$ after noise cancellation as the following affine set:
\begin{equation}
\label{affine set}
\Sigma (\mathcal D)=\{A\in\mathbb{R}^{n\times n}|\;R(\mathcal D)=QAP(\mathcal D)\}.
\end{equation}

Previous studies have discussed system-theoretic properties such as strong controllability and strong observability for affine sets of this form. In particular, when data are given, conditions are provided for any model $A$ contained within the affine set described by that data to possess desired properties, which are termed data informativity. 

In this paper, we focus particularly on observability.  To that end, we provide the following fundamental definition of strong observability. Here, $w_k$ in (\ref{plant}) represents noise that corrupts the data and is unrelated to the system's inherent properties; therefore, when defining observability, we set $w_k\equiv0$. Furthermore, for the system (\ref{plant}) without noise, let $y(k,x_0,u)$ denote the system output when the input $u$ is applied with the initial state $x_0$.
\begin{dfn}
The system (\ref{plant}) is called strongly observable if for each $x_0\in\mathbb{R}^n$ and input sequence $u$, it holds that $y(k, x_0, u) = 0$ for all $k\in\mathbb{Z}_+$ implies that $x_0 = 0$.
\end{dfn}

When there are no input terms in (\ref{plant}), this definition coinsides with the usual observability. Furthermore, we introduce the invariant subspace corresponding to the usual unobservable subspace.
\begin{dfn}
Consider the system (\ref{plant}). We define the subspace $\mathcal V^*(A,B,C,D)\subseteq \mathbb{R}^n$ called the weakly unobservable subspace as 
\begin{equation}
\label{weakly unobservable subspace}
\begin{aligned}[b]
\mathcal V^*&(A,B,C,D)\\:=&\{x_0\in\mathbb{R}^n|\;\exists u \;\mathrm{s.t.}\;y(k,x_0,u)=0,\;k\in\mathbb{Z}_+\}.
\end{aligned}
\end{equation}
\end{dfn}
\vspace{1mm}

Clearly, the system (\ref{plant}) being strongly observable is equivalent to $\mathcal V^*(A,B,C,D)=\{0\}$.

We now introduce a geometric approach based on
invariant subspaces from \cite{Eising}. As we will see, 
this becomes fundamental 
for our study on data-driven attacks 
targeting invariant subspaces. We say a subspace $\mathcal V\subseteq \mathbb{R}^n$ is output-nulling controlled invariant if
\begin{equation}
\label{output-nulling controlled invariant}
\begin{bmatrix}
    A\\C
\end{bmatrix}\mathcal V\subseteq \mathcal V\times\{0\}+\mathrm{im}\begin{bmatrix}
    B \\ D
\end{bmatrix}.
\end{equation}
It is known that the largest subspace satisfying this inclusion relation coincides with the weakly unobservable subspace $\mathcal V^*(A,B,C,D)$ defined in (\ref{weakly unobservable subspace}) (see \cite{Trentelman}, Theorem 7.10). That is, $\mathcal V^*(A,B,C,D)$ includes all $\mathcal V$ satisfying (\ref{output-nulling controlled invariant}). 

We can transform the relation in (\ref{output-nulling controlled invariant}) based on a data-driven approach. 
As preparation, we interpret the inclusion on the coefficient side: we introduce a subspace $\mathcal J(\mathcal D)\subseteq \mathbb R^{T}$ whose elements represent linear combination coefficients that can generate weakly unobservable state trajectories from the data $P(\mathcal D) = X_-$. This expression is based on a behavioral perspective, such as Willems' fundamental lemma, which states that a trajectory of the system can be generated by a linear combination of trajectory data \cite{Willems, Multiple datasets}. Therefore, we refer to $\mathbb{R}^T$ as the \emph{coefficient space} and $\mathcal J(\mathcal D)$ as the weakly unobservable coefficient space described by the data $\mathcal D$. 

Although the parameter $A$ in (\ref{output-nulling controlled invariant}) is currently unknown, using the fact that it is an element of the affine model set $\Sigma (\mathcal D)$, we can derive the data-driven representation
\begin{equation}
\label{data driven coefficient space}
\begin{bmatrix}
    R(\mathcal D)\\CP(\mathcal D)
\end{bmatrix}\mathcal J(\mathcal D)\subseteq QP(\mathcal D)\mathcal J(\mathcal D)\times\{0\}+\mathrm{im}\begin{bmatrix}
    QB \\ D
\end{bmatrix}.
\end{equation}
Furthermore, we define the following subspaces.
\begin{dfn}
\label{maximum weakly unobservable coefficient space}
We call the largest subspace $\mathcal J^*(\mathcal D)\subseteq\mathbb{R}^T$ satisfying the inclusion (\ref{data driven coefficient space}) the maximum weakly unobservable coefficient space.
\end{dfn}
\begin{dfn}
Given data $\mathcal D$ and $\mathcal J^*(\mathcal D)$, we refer to the subspace $\mathcal V^*(\mathcal D):=P(\mathcal D)\mathcal J^*(\mathcal D)\subseteq\mathbb{R}^n$ as the data-driven weakly unobservable subspace.
\end{dfn}

In the spirit of the general data-informativity framework \cite{Informativity}, we may also express informativity for strong observability in an abstract set-inclusion form. We define the entire set of strongly observable models as
\begin{equation*}
\Sigma_{\mathrm{SO}} := \{A\in\mathbb{R}^{n\times n}\mid (A,B,C,D)\ \text{is strongly observable}\}.
\end{equation*}
Then, we define primitive informativity as follows.
\begin{dfn}
The data $\mathcal D$ are said to be informative for strong observability if $\Sigma(\mathcal D)\subseteq \Sigma_{\mathrm{SO}}$.
\end{dfn}

The following result characterizes informativity using the geometric condition in (\ref{data driven coefficient space}).
\begin{lmm}[\hspace{-0.01pt}\cite{Eising}, Theorem 22]
\label{informativity under geometric condition}
The data $\mathcal D$ are informative for strong observability if and only if $C^{-1}\mathrm{im}\,D\subseteq\mathrm{im}\,P(\mathcal D)$ and $\mathcal J^*(\mathcal D)\subseteq \mathrm{ker}\,P(\mathcal D)$.
\end{lmm}

\subsection{Problem formulation}
In this paper, we consider a problem from the perspective of an attacker whose goal is to process the data $\mathcal D$ so that the modified data $\tilde{\mathcal D}$ becomes not informative for strong observability. Also, \emph{stealthy} attacks are important, since the transformation should not destroy the original information in the data (e.g., rank and the original coefficient space $\mathcal J^*(\mathcal D)$). More specifically, the class of attacks that we consider is represented by the linear transformation $\Phi_a\in\mathbb R^{(2n+m+p)\times(2n+m+p)}$ given by 
\begin{equation*}
\Phi_a:=\textrm{blk-diag}(\Phi_-^X,\Phi_+^X,\Phi_-^U,\Phi_-^Y),
\end{equation*}
which will be applied to the original data as follows:
\begin{equation}
\label{data transformation}
\tilde{\mathcal D}:=\Phi_a\mathcal D.
\end{equation}
As a consequence, the attacked data are given by
\begin{equation*}
\begin{split}
\tilde{X}_-&:=\Phi_-^X X_-,\quad
\tilde{X}_+:=\Phi_+^X X_+,\\
\tilde{U}_-&:=\Phi_-^U U_-,\quad
\tilde{Y}_-:=\Phi_-^Y Y_-.
\end{split}
\end{equation*}
We impose as part of the design requirement that the linear transformation $\Phi_a$ used for the attack is nonsingular. This is because attacks that reduce the rank of the original data would be trivial, and they will be immediately detected by the system operator when inputs with persistency of excitation (PE) properties are applied \cite{DePersis PE}.

In view of Lemma \ref{informativity under geometric condition}, the attacked data $\tilde{\mathcal D}$ is not informative for strong observability if $\mathcal J^*(\tilde{\mathcal D})\nsubseteq\mathrm{ker}\,P(\tilde{\mathcal D})$. This can be written as
\begin{equation*}
\exists v\in\mathcal J^*(\tilde{\mathcal D}),\;v\neq0\;\mathrm{s.t.} \;P(\tilde{\mathcal D})v\neq0
\end{equation*}
as an equivalent condition. This means that there exists a coefficient vector $v\in \mathcal J^*(\tilde{\mathcal D})$ that generates a nonzero state vector through the linear map $P(\tilde{\mathcal D})$. In particular, letting $\tilde{x}_0:=P(\tilde{\mathcal D})v$, we have $\tilde{x}_0\in\mathcal V^*(\tilde{\mathcal D})=P(\tilde{\mathcal D})\mathcal J^*(\tilde{\mathcal D})$ and $\tilde{x}_0\neq0$. Therefore, the attacker's goal is to embed such a state into the data-driven weakly unobservable subspace after data transformation. 

We also consider the problem of minimizing the data perturbation. This is done not only to reduce detectability via statistical tests as shown in \cite{Russo}, but also to characterize the minimal perturbation required to destroy the informativity.

Based on the above discussion, the problem we aim to solve can be stated as follows.
\begin{prb}
\label{problem 1}
For the system in (\ref{plant}) and the given original data $\mathcal D$, under Assumption \ref{assumption A_dat}, design an invertible linear map $\Phi_a:\mathbb D\to\mathbb D$ as in (\ref{data transformation}) so that $\tilde{x}_0\in\mathcal V^*(\tilde{\mathcal D})$ holds. Furthermore, find the transformation $\Phi_a^*$ among those satisfying the above conditions that minimizes $\|\Phi_a^*\mathcal D-\mathcal D\|_\mathrm{F}$, where $\|\cdot\|_\mathrm{F}$ denotes the Frobenius norm.
\end{prb}

\section{Attack design}
In this section, we provide an algorithm for designing a specific transformation that solves Problem \ref{problem 1}. In particular, the malicious state $\tilde x_0$ is treated as fixed throughout this section, and the attacker designs an invertible linear map $\Phi_a$ so that the resulting data become not informative for strong observability by injecting the mode $(\tilde \lambda,\tilde x_0)$. To formalize the attacker’s capability, we make the following assumption.
\begin{asm}
\label{assumption of attacker}
The attacker has a virtual model $A_{\mathrm{mal}}\in\mathbb{R}^{n\times n}$ and its eigenpair $(\tilde{\lambda},\tilde{x}_0)$ such that the system $(A_{\mathrm{mal}}, B,C,D)$ is weakly unobservable and $A_{\mathrm{mal}}\tilde{x}_0=\tilde{\lambda} \tilde{x}_0$.
\end{asm}

\subsection{Construction of malicious trajectories}
Under Assumption \ref{assumption of attacker}, we consider embedding the malicious state $\tilde{x}_0$ into $\mathcal V^*(\tilde{\mathcal D})$ based on the geometric condition (\ref{data driven coefficient space}). Notice here that $\mathcal J^*(\mathcal D)$ is the coefficient space that can generate $x_0=P(\mathcal D)\eta$ with $\eta\in\mathcal J^*(\mathcal D)$, which results in $x_0\in\mathcal V^*(\mathcal D)$. Thus, for the attack design, it is sufficient to find  $v\in\mathcal J^*(\tilde{\mathcal D})$ satisfying $\tilde{x}_0=P(\tilde{\mathcal D})v$ from the perspective of the coefficient space. 

For the attacked data $\tilde{\mathcal D}$, the geometric condition in (\ref{data driven coefficient space}) is written as
\begin{equation}
\label{data driven attacked coefficient space}
\begin{bmatrix}
    R(\tilde{\mathcal D})\\CP(\tilde{\mathcal D})
\end{bmatrix}\mathcal J^*(\tilde{\mathcal D})\subseteq QP(\tilde{\mathcal D})\mathcal J^*(\tilde{\mathcal D})\times\{0\}+\mathrm{im}\begin{bmatrix}
    QB \\ D
\end{bmatrix}.
\end{equation}
On the other hand, since $A_{\mathrm{mal}}$ should be included in the affine set after the attack, $R(\tilde{\mathcal D})=QA_{\mathrm{mal}}P(\tilde{\mathcal D})$ must hold. Thus, the attack succeeds if
\begin{equation*}
\exists v\in\mathcal J^*(\tilde{\mathcal D}),\;v\neq0\;\text{s.t.}\;R(\tilde{\mathcal D})v=QA_{\mathrm{mal}}P(\tilde{\mathcal D})v
\end{equation*}
as an element of the left-hand side set of (\ref{data driven attacked coefficient space}). Therefore, for the inclusion in (\ref{data driven attacked coefficient space}) to hold, it is necessary and sufficient that there exist vectors $v\neq0$ and $\beta$ satisfying
\begin{equation}
\label{equation of attack design}
\begin{bmatrix}
R(\tilde{\mathcal D})\\CP(\tilde{\mathcal D})
\end{bmatrix}v
=
\begin{bmatrix}
QA_{\mathrm{mal}}P(\tilde{\mathcal D})
\\0
\end{bmatrix}v
+
\begin{bmatrix}
QB\\
D
\end{bmatrix}\beta.
\end{equation}
Expanding (\ref{equation of attack design}) yields
\begin{equation}
\label{separated equation}
\begin{split}
R(\tilde{\mathcal D})v&=QA_{\mathrm{mal}}P(\tilde{\mathcal D})v+QB\beta,\\
CP(\tilde{\mathcal D})v&=D\beta.
\end{split}
\end{equation}
By (\ref{affine parameters}) and substituting the transformed data into (\ref{separated equation}), we obtain
\begin{equation}
\label{expanded equation}
\begin{split}
M(\tilde{X}_+v-B\tilde{U}_-v)+&N(\tilde{Y}_-v-C\tilde{X}_-v-D\tilde{U}_-v)\\&=M(A_{\mathrm{mal}}\tilde{X}_-v+B\beta),\\
C\tilde{X}_-v&=D\beta.
\end{split}
\end{equation}

We now introduce the target vectors that the attacker must find: First, from (\ref{affine parameters}), we have $\tilde{x}_0 = P(\tilde{\mathcal D})v = \tilde{X}_-v$, which is the first target vector. There are three more as follows:
\begin{equation}
\label{relation between attacked data and targets}
\tilde{x}_1:=\tilde{X}_+v,\;\tilde{u}_0:=\tilde{U}_-v,\;\tilde{y}_0:=\tilde{Y}_-v.
\end{equation}
By substituting them into (\ref{expanded equation}), we obtain
\begin{equation}
\label{expanded equation 2}
\begin{split}
M(\tilde{x}_1-B\tilde{u}_0)+&N(\tilde{y}_0-C\tilde{x}_0-D\tilde{u}_0)\\&=M(A_{\mathrm{mal}}\tilde{x}_0+B\beta),\\
C\tilde{x}_0&=D\beta.
\end{split}
\end{equation}

In this study, to simplify the problem, we consider the special case with $\beta=0$\footnote{If $\beta\neq0$ is allowed, the constraint $\tilde{x}_0\in C^{-1}\mathrm{im}\,D$ is imposed on the initial state chosen by the attacker, which complicates the attack design. On the other hand, setting $\beta=0$ simplifies the problem greatly: $C\tilde{x}_0=0$, meaning $\tilde{x}_0$ can just be taken from the kernel of $C$.}. As one solution satisfying the above equations, by comparing the coefficients of $M$ and $N$ on both sides of (\ref{expanded equation 2}), we obtain the following conditions to be satisfied by the target vectors:
\begin{equation}
\label{design vectors}
\left\{
\begin{split}
\tilde{x}_0&=\Phi_-^XX_-v,\\
\tilde{u}_0&=\Phi_-^UU_-v,\\
\tilde{x}_1&=\Phi_+^XX_+v=\tilde{\lambda}\tilde{x}_0+B\tilde{u}_0,\\
\tilde{y}_0&=\Phi_-^YY_-v=D\tilde{u}_0,\\
C\tilde{x}_0&=0,
\end{split}
\right.
\end{equation}
where we applied $A_{\mathrm{mal}}\tilde{x}_0=\tilde{\lambda}\tilde{x}_0$. From (\ref{design vectors}), the target vectors can be interpreted as one-step virtual trajectories obtained from the malicious model $(A_{\mathrm{mal}}, B, C, D)$ specified by the attacker. Here, from (\ref{design vectors}), we see that the attack design procedure requires the initial input $\tilde{u}_0$ in addition to the eigenpair $(\tilde{\lambda},\tilde{x}_0)$ to be specified first.

\subsection{Attack designs and analysis of their feasibility}
We proceed to consider the problem of finding a linear map $\Phi_Z$ for each data $Z\in\{X_-,X_+,U_-,Y_-\}$ that generates the target vectors in (\ref{design vectors}). Specifically, we design each block $\Phi_Z$ so that it maps the data-generated vector $Z v$ to the prescribed target $z_{\mathrm{tar}}$ in (\ref{design vectors}), while leaving the original weakly unobservable coefficient directions $\mathcal J^*(\mathcal D)$ unchanged. The following lemma provides a parameterization of such a map $\Phi_Z$ via a vector $\xi_Z$.
\begin{lmm}
\label{attack design}
Suppose that the data $Z\in\mathbb{R}^{m_Z\times T}$, the vector $v\in\mathbb{R}^T\setminus\{0\}$, and the target vector $z_\mathrm{tar}\in\mathbb{R}^{m_Z}$ are given. For any $\xi_Z\in\mathbb{R}^{m_Z}$ such that $\xi_Z^\mathsf{T}Zv=1$ and $\xi_Z^\mathsf{T}z_\mathrm{tar}\neq0$, let $\Phi_Z:=I_{m_Z}+(z_\mathrm{tar}-Zv)\xi_Z^\mathsf{T}$. Then, the following statements hold:
\begin{enumerate}
\item $\Phi_Z$ is nonsingular.

\item $\Phi_ZZv=z_\mathrm{tar}$.

\item For any $w\in\mathrm{ker}\,(\xi_Z^\mathsf{T}Z)$, it holds $\Phi_ZZw=Zw$.
\end{enumerate}
\end{lmm}
\vspace{1mm}

This lemma suggests us how to construct maps to transform the data.
By (ii), the map $\Phi_Z$ generates attacks in the
$v$ direction, while (iii) ensures that the data appears
unchanged in the original coefficient direction. The
latter property enhances stealthiness because it makes it impossible to distinguish between the original data and the attacked data in the $w$ direction. 

The attack design procedure is shown in Algorithm \ref{alg:attack}.
This algorithm leads us to the nonsingular transformation $\Phi_a=\textrm{blk-diag}(\Phi_-^X,\Phi_+^X,\Phi_-^U,\Phi_-^Y)$ with
\begin{equation}
\label{transformations}
\begin{split}
&\Phi_-^X:=I_n+(\tilde{x}_0-X_-v)\xi_{X_-}^{\mathsf{T}},\\
&\Phi_+^X:=I_n+(\tilde{x}_1-X_+v)\xi_{X_+}^{\mathsf{T}},\\
&\Phi_-^U:=I_m+(\tilde{u}_0-U_-v)\xi_{U_-}^{\mathsf{T}},\\
&\Phi_-^Y:=I_p+(\tilde{y}_0-Y_-v)\xi_{Y_-}^{\mathsf{T}}.\\
\end{split}
\end{equation}

We hence arrive at the following approach for the attacker to design a specific map: First, choose the attack coefficient direction $v$ not belonging to $\mathcal J^*(\mathcal D)$. Second, take $\tilde{x}_0$ and $\tilde{u}_0$, and use $A_{\mathrm{mal}}$, $B$, and $D$ to compute the target vectors $\tilde{x}_1=A_{\mathrm{mal}}\tilde{x}_0+B\tilde{u}_0=\tilde{\lambda}\tilde{x}_0+B\tilde{u}_0$ and $\tilde{y}_0=D\tilde{u}_0$. Third, select the corresponding $\xi_Z$ for each of these target vectors and the attack coefficient direction $v$. This can be calculated as follows: For the data $Z$, define the subspace 
\begin{equation*}
\Pi_\mathcal O(Z):=Z\mathcal J^*(\mathcal D). 
\end{equation*}
It collects the outputs generated by the original weakly unobservable coefficient space. Take a nonzero normal vector $u_Z$ in its orthogonal complement $\Pi_\mathcal O(Z)^\perp$, and parametrize $\xi_Z=(1/(u_Z^\mathsf{T}Zv))u_Z$. To ensure that this vector $\xi_Z$ is not orthogonal to $z_\mathrm{tar}$, we introduce the following assumption for the target vector $z_\mathrm{tar}$\footnote{
This assumption is not restrictive. In particular, if $\Pi_\mathcal O(Z)$ is a proper subspace, it has Lebesgue measure zero. Hence, $z_\mathrm{tar}\notin\Pi_\mathcal O(Z)$ generally holds. For the reason, this assumption will not be explicitly included in the design conditions in the further discussion.}.
\begin{asm}
For each pair $(z_\mathrm{tar},Z)\in \{(\tilde{x}_0,X_-),(\tilde{x}_1,\\X_+),(\tilde{u}_0,U_-),(\tilde{y}_0,Y_-)\}$, it holds that $z_\mathrm{tar}\notin\Pi_\mathcal O(Z)$.
\end{asm}

The question is whether there exists a vector $\xi_Z$ that satisfies all the design requirements of Lemma~\ref{attack design}. The following proposition  provides a solution to this question. In particular, several equivalent sufficient conditions are shown for Algorithm \ref{alg:attack} to be feasible.

\begin{prp}
\label{prp 1}
Let $\mathcal D=(X_-,X_+,U_-,Y_-)$ be the given data, and let $\mathcal J^*(\mathcal D)\subseteq\mathbb{R}^T$ be the maximum weakly unobservable coefficient space. Then, for each $Z\in\{ X_-,X_+, U_-,Y_- \}$, the following statements are equivalent:
\begin{enumerate}
\item (Feasibility of the design conditions)
There exist $v\in \mathcal J^*(\mathcal D)^\perp\setminus\{0\}$ and $\xi_Z\in\mathbb{R}^{m_Z}$ such that 
\begin{equation}
\label{design condition}
\xi_Z^\mathsf{T}Zw=0,\;\forall w\in\,\mathcal J^*(\mathcal D),\;\xi_Z^\mathsf{T}Zv=1,
\end{equation}
where $m_Z$ is the number of rows of $Z$.

\item (Geometric condition on the coefficient vector)
There exists $v\in \mathcal J^*(\mathcal D)^\perp\setminus\{0\}$ such that
\begin{equation}
\label{geometric condition}
v\notin Z^{-1}\Pi_\mathcal O(Z).
\end{equation}

\item (Dimensional condition on the coefficient space)
The following condition holds:
\begin{equation}
\label{dimensional condition}
\mathrm{dim}\,\Pi_\mathcal O(Z)<\mathrm{rank}\,Z.
\end{equation}
\end{enumerate}
Moreover, if any of (i)--(iii) holds, then the linear map $\Phi_a:\mathbb D\to\mathbb D$  in (\ref{transformations}) exists and that solves Problem \ref{problem 1}.
\end{prp}

Condition~(i) is nothing but the design requirement in Lemma~2. Condition~(ii) concerns the \emph{designability} of $\xi_Z$ through the free vector $u_Z\in\Pi_\mathcal O(Z)^\perp$ : since $\xi_Z$ is realized as $\xi_Z=(1/(u_Z^\mathsf{T}Zv))u_Z$, it guarantees that one can pick $u_Z$ so that $\Phi_Z$ satisfies the invariance constraints on $Z\mathcal J^*(\mathcal D)$ while still meeting the injection requirement. Finally, condition (iii) means that $\Pi_\mathcal O(Z)=Z\mathcal J^*(\mathcal D)$ does not fill the whole image $\mathrm{im}\,(Z)$. Thus, for each $Z$, there remains a direction outside $\Pi_\mathcal O(Z)$, which makes it possible to search for a common $v\in \mathcal J^*(\mathcal D)^\perp\setminus\{0\}$ such that $v\notin Z^{-1}\Pi_\mathcal O(Z)$. Unlike (i) and (ii), it contains no design variables and can be checked directly from the data.

\begin{algorithm}[tb]
\caption{Data-driven structural cyber attack}
\label{alg:attack}
\begin{algorithmic}[1]
\REQUIRE Data $\mathcal D = (X_-, X_+, U_-, Y_-)$;
         the maximum weakly unobservable coefficient space $\mathcal J^*(\mathcal D)\subseteq\mathbb{R}^T$;
         target vector $z_{\mathrm{tar}}$ for each $Z \in \{X_-,X_+,U_-,Y_-\}$.
\ENSURE Attacked data $\tilde{\mathcal {D}} = (\tilde X_-, \tilde X_+, \tilde U_-, \tilde Y_-)$.

\STATE Select $v \in \mathcal J^*(\mathcal D)^\perp$ with $v \notin Z^{-1}\Pi_\mathcal O(Z)$ for all $Z$.
\STATE \textbf{Construction of $\xi_Z$ and $\Phi_Z$:}
  \FOR{each $Z \in \{X_-,X_+,U_-,Y_-\}$}
    \STATE Choose a nonzero normal vector $u_Z \in \Pi_\mathcal O(Z)^\perp$.
    \STATE Set $\xi_Z := (1/(u_Z^\top Z v))u_Z$.
    \STATE Set the attack matrix $\Phi_Z := I + (z_{\mathrm{tar}} - Z v)\,\xi_Z^\mathsf{T}$,
           and then the attacked block $\tilde Z := \Phi_Z Z$.
  \ENDFOR

\STATE \textbf{Output:}
  Return $\tilde {\mathcal {D}} = (\tilde X_-, \tilde X_+, \tilde U_-, \tilde Y_-)$.
\end{algorithmic}
\end{algorithm}

\subsection{Subspace analysis under attacks}
We are now ready to state our first main result. It analyzes how the weakly unobservable coefficient space $\mathcal J^*(\mathcal D)$ and the model set $\Sigma(\mathcal D)$ transform under attacks designed by Algorithm \ref{alg:attack}.
\begin{thm}
\label{thm 1}
Suppose that any of the conditions (i)--(iii) in Proposition \ref{prp 1} holds, and the data $\mathcal D$ is transformed to $\tilde{\mathcal D}$ according to Algorithm \ref{alg:attack}. Then, the following statements hold:

\begin{enumerate}
\item $
\mathcal J^*(\mathcal D)\oplus\mathrm{span}\{v\}\subseteq \mathcal J^*(\tilde{\mathcal D})$.

\item If $\Sigma(\tilde{\mathcal D})\neq\varnothing$, then
\begin{equation*}
\Sigma(\tilde{\mathcal D})\cap\{A_\mathrm{mal}\in\mathbb{R}^{n\times n}|\;A_\mathrm{mal}\tilde{x}_0=\tilde{\lambda}\tilde{x}_0\}\neq\varnothing.
\end{equation*}

\item $\tilde{\mathcal D}$ is not informative for strong observability.
\end{enumerate}
\end{thm}
\vspace{1mm}
\proof
(i) First, we show that $\mathrm{span}\{v\}\subseteq \mathcal J^*(\tilde{\mathcal D})$. Regarding $R(\tilde{
\mathcal D})$ after attack, substituting the attacked data and attack direction $v$ into (\ref{affine parameters}) and using (\ref{relation between attacked data and targets}) yields
\begin{equation*}
\begin{aligned}[b]
R(\tilde{\mathcal D})v&=M(\tilde{X}_+v-B\tilde{U}_-v)+N(\tilde{Y}_-v-C\tilde{X}_-v-D\tilde{U}_-v)\\
&=M(\tilde{x}_1-B\tilde{u}_0)+N(\tilde{y}_0-C\tilde{x}_0-D\tilde{u}_0).
\end{aligned}
\end{equation*}
Furthermore, from (\ref{design vectors}), we have $R(\tilde{\mathcal D})v=M\tilde{\lambda}\tilde{x}_0$. Here, using the fact that $M=Q$ and $\tilde{x}_0=\tilde{X}_-v=P(\tilde{\mathcal D})v$,
\begin{equation}
\label{Rv=lambdaQPv}
R(\tilde{\mathcal D})v=\tilde{\lambda}QP(\tilde{\mathcal D})v.
\end{equation}
Also, $CP(\tilde{\mathcal D})v=C\tilde{X}_-v=C\tilde{x}_0=0$.
To summarize the above,
\begin{equation*}
\begin{bmatrix}
R(\tilde{\mathcal D})\\CP(\tilde{\mathcal D})
\end{bmatrix}v=
\begin{bmatrix}
QP(\tilde{\mathcal D})\\0
\end{bmatrix}(\tilde{\lambda}v)+
\begin{bmatrix}
QB\\D
\end{bmatrix}0.
\end{equation*}
Therefore, 
\begin{equation*}
\begin{bmatrix}
R(\tilde{\mathcal D})\\CP(\tilde{\mathcal D})
\end{bmatrix}\mathrm{span}\{v\}\subseteq
QP(\tilde{\mathcal D})\mathrm{span}\{v\}\times\{0\}+\mathrm{im}
\begin{bmatrix}
QB\\D
\end{bmatrix}.
\end{equation*}
This shows that $\mathrm{span}\{v\}$ satisfies (\ref{data driven attacked coefficient space}). Moreover, since $\mathcal J^*(\tilde{\mathcal D})$ is the largest subspace satisfying (\ref{data driven attacked coefficient space}), we have
\begin{equation}
\label{span in J}
\mathrm{span}\{v\}\subseteq\mathcal J^*(\tilde{\mathcal D}).
\end{equation}

Next, we show $\mathcal J^*(\mathcal D)\subseteq \mathcal J^*(\tilde {\mathcal D})$. Recall that $\mathcal J^*(\mathcal D)$ is defined as the largest subspace satisfying the inclusion (\ref{data driven coefficient space}) for the original data $\mathcal D$. 
Fix $Z\in\{X_-,X_+,U_-,Y_-\}$.
By Proposition \ref{prp 1}, the design condition in (\ref{design condition}) ensures that
\begin{equation*}
\xi_Z^\mathsf{T}Z w = 0,\quad \forall w\in\mathcal  J_*(\mathcal D). 
\end{equation*}
Hence, for every $w\in \mathcal J^*(\mathcal D)$, we have $w\in\mathrm{ker}\,(\xi_Z^\mathsf{T} Z)$.
Now by Lemma \ref{attack design}--(iii), it follows that $\Phi_Z Z w = Z w$, and therefore $\tilde Z w = \Phi_Z Z w = Z w$.
Since this holds for any data $Z$, we conclude that for all $w\in \mathcal J^*(\mathcal D)$,
\begin{equation*}
\tilde X_- w = X_- w,\, \tilde X_+ w = X_+ w,\, \tilde U_- w = U_- w,\, \tilde Y_- w = Y_- w. 
\end{equation*}
Consequently, all data-dependent linear maps built from these data matrices coincide on $\mathcal J^*(\mathcal D)$:
\begin{equation*}
P(\tilde {\mathcal D})w = P(\mathcal D)w,\;R(\tilde {\mathcal D})w = R(\mathcal D)w,
\quad \forall w\in \mathcal J^*(\mathcal D).
\end{equation*}
Furthermore, the subspace images coincide as
\begin{equation*}
P(\tilde {\mathcal D})\mathcal J^*(\mathcal D) = P(\mathcal D)\mathcal J^*(\mathcal D).
\end{equation*}
Using the equalities above and (\ref{data driven coefficient space}), we can show
\begin{equation*}
\begin{split}
\begin{bmatrix}R(\tilde {\mathcal D})\\ CP(\tilde {\mathcal D})\end{bmatrix}\mathcal J^*(\mathcal D)&=\begin{bmatrix}R(\mathcal D)\\ CP(\mathcal D)\end{bmatrix}\mathcal J^*(\mathcal D)
\\&\subseteq
QP(\mathcal D)\mathcal J^*(\mathcal D)\times\{0\} + \mathrm{im}\begin{bmatrix}QB\\ D\end{bmatrix}\\&=QP(\tilde {\mathcal D})\mathcal J^*(\mathcal D)\times\{0\} + \mathrm{im}\begin{bmatrix}QB\\ D\end{bmatrix}.
\end{split}
\end{equation*}
Therefore, $\mathcal J^*(\mathcal D)$ is a feasible subspace for the inclusion (\ref{data driven coefficient space}) associated with
the attacked data $\tilde {\mathcal D}$. By the maximality of $\mathcal J^*(\tilde {\mathcal D})$,
it must contain any feasible subspace, and in particular,
\begin{equation}
\label{J in J tilde}
\mathcal J^*(\mathcal D)\subseteq \mathcal J^*(\tilde {\mathcal D}).  
\end{equation}

We have now established (\ref{span in J}) and (\ref{J in J tilde}). It then follows
\begin{equation}
\label{span + J in J tilde}
\mathrm{span}\{v\}+\mathcal J^*(\mathcal D)\subseteq \mathcal J^*(\tilde {\mathcal D}).
\end{equation}
Furthermore, $v\in\mathcal J^*(\mathcal D)^\perp$ by the design in Algorithm \ref{alg:attack}, and hence we have $\mathrm{span}\{v\}\cap\mathcal J^*(\mathcal D)=\{0\}$. Thus, the sum in (\ref{span + J in J tilde}) becomes a direct sum. 

(ii) Since $\Sigma(\tilde{\mathcal D})\neq\varnothing$, arbitrarily select one model $A_0\in \Sigma(\tilde{\mathcal D})$. From the definition of the affine model set in (\ref{affine set}), we have $R(\tilde{\mathcal D})=QA_0P(\tilde{\mathcal D})$. By multiplying both sides by the vector $v$ from the right, $R(\tilde{\mathcal D})v=QA_0P(\tilde{\mathcal D})v=QA_0\tilde{x}_0$. On the other hand, from (\ref{Rv=lambdaQPv}), we obtain $R(\tilde{\mathcal D})v=\tilde{\lambda} Q\tilde{x}_0$. Using these relations, we see that $Q(\tilde{\lambda}I_n-A_0)\tilde{x}_0=0$.
Here, let $w := (\tilde{\lambda}I_n-A_0)\tilde{x}_0$. Then, $Qw = 0$. Next, since $\tilde{x}_0\neq0$, define the matrix $\Delta A:=w\tilde{x}_0^\mathsf{T}/{\|\tilde{x}_0\|_2^2}$ and set $A^*:=A_0+\Delta A$. Then, we have
\begin{equation*}
\begin{split}
A^*\tilde{x}_0=A_0\tilde{x}_0+\Delta A\tilde{x}_0&=A_0\tilde{x}_0+\frac{w}{\|\tilde{x}_0\|_2^2}\tilde{x}_0^\mathsf{T}\tilde{x}_0\\&=A_0\tilde{x}_0+w=\tilde{\lambda}\tilde{x}_0.
\end{split}
\end{equation*}
Therefore, $A^*$ has an eigenpair $(\tilde{\lambda},\tilde
x{_0})$. Finally, we show that $A^*\in\Sigma (\tilde{\mathcal D})$. Since $Qw=0$, we have $Q\Delta A=0$. Thus,
\begin{equation*}
QA^*P(\tilde{\mathcal D})=Q(A_0+\Delta A)P(\tilde{\mathcal D})=QA_0P(\tilde{\mathcal D})=R(\tilde{\mathcal D}).
\end{equation*}
Therefore, $A^*$ is included in the affine model set $\Sigma (\tilde{\mathcal D})$.

(iii) The proof of this part is straightforward. By Algorithm \ref{alg:attack} and the result of (i),
\begin{equation*}
P(\tilde{\mathcal D})v=\Phi_-^XX_-v=\tilde{x}_0\neq0,\;v\in\mathrm{span}\{v\}\subseteq\mathcal J^*(\tilde{\mathcal D}).
\end{equation*}
Hence, it follows that $\mathcal J^*(\tilde{\mathcal D})\nsubseteq\mathrm{ker}\,P(\tilde{\mathcal D})$.
\endproof

This theorem not only guarantees that the attacker's specified malicious state is contained within the subspace, but also demonstrates that unobservable models are included in the model set as described in the data $\tilde{\mathcal D}$. This also suggests that it may render quadratic stabilization impossible for the entire model set in the conventional data informativity design problem. That is, injecting an undetectable or unstabilizable model as a malicious model makes it impossible to simultaneously stabilize all models that can be explained by the data.

\section{Minimum norm attacks}
In this section, we discuss the design of attacks that destroy strong observability with minimal data perturbation. In contrast to Section~III where the target eigenpair is prescribed a priori, the minimum-norm formulation jointly optimizes over $(\tilde \lambda,\tilde x_0)$ and the attack, thereby identifying the eigenpair whose injection requires the least perturbation. This section describes the formulation of the minimization problem for finding such optimal eigenpairs, and the relation between the minimum norm and the model set $\Sigma(\mathcal D)$.

Suppose that the original data $\mathcal D$ satisfies the dimensional condition (\ref{dimensional condition}) in Proposition \ref{prp 1}. By Algorithm \ref{alg:attack}, for each data $Z\in\{X_-,X_+,U_-,Y_-\}$, the attacked data $\tilde{Z}$ is given as
\begin{equation}
\tilde{Z}=Z+(z_\mathrm{tar}-Zv)\xi_Z^\mathsf{T}Z.
\end{equation}
Here, we define the data perturbation $\Delta_Z$ as
\begin{equation}
\Delta_Z:=\tilde{Z}-Z=(z_\mathrm{tar}-Zv)\xi_Z^\mathsf{T}Z.
\end{equation}
That is, for the stacked data $\mathcal D$,
\begin{equation}
\label{data perturbation}
\tilde{\mathcal D}=\mathcal D+\Delta\mathcal D,\;\Delta\mathcal D:=\begin{bmatrix}
    \Delta_{X_-}^\mathsf{T} & \Delta_{X_+}^\mathsf{T} & \Delta_{U_-}^\mathsf{T} & \Delta_{Y_-}^\mathsf{T}
\end{bmatrix}^\mathsf{T}.
\end{equation}
Now, we wish to design an attack that minimizes this data perturbation across all data. To do so, we define the following objective function:
\begin{equation}
\label{objective}
\mathcal L(\{z_\mathrm{tar}\},v,\{\xi_Z\}):=\|\Delta\mathcal D\|_\mathrm{F}^2=\sum_{Z\in\{X_-,X_+,U_-,Y_-\}}\|\Delta_Z\|_\mathrm{F}^2.
\end{equation}
Since the Frobenius norm of the outer product is the product of the 2-norms of the two vectors, we have
\begin{equation}
\label{perturbation}
\|\Delta_Z\|_\mathrm{F}^2=\|(z_\mathrm{tar}-Zv)\xi_Z^\mathsf{T}Z\|_\mathrm{F}^2=\|z_\mathrm{tar}-Zv\|_2^2\|\xi_Z^\mathsf{T}Z\|_2^2.
\end{equation}
Here, to proceed with the minimal perturbation analysis, we consider setting $\Delta_{X_-}=0$, $\Delta_{U_-}=0$, and $\Delta_{Y_-}=0$, thereby assigning the effect of data perturbation to $X_+$. It is sufficient to set $\Phi_-^{X} = I_n$, $\Phi_-^{U} = I_m$, and $\Phi_-^Y= I_p$ and take $\tilde x_0 := X_- v$, $\tilde u_0 := U_- v$, and $\tilde y_0 := Y_- v$ for the chosen coefficient direction $v$. Furthermore, using the relation $\tilde{x}_1=\tilde{\lambda}\tilde{x}_0+B\tilde{u}_0$ among the target vectors, the objective function (\ref{objective}) with (\ref{perturbation}) becomes as follows:
\begin{equation}
\label{objective function 2}
\mathcal L(\lambda,v,\xi)=\|(\lambda X_--X_++BU_-)v\|_2^2\|\xi^\mathsf{T} X_+\|_2^2,
\end{equation}
where $\lambda:=\tilde{\lambda}$ and $\xi:=\xi_{X_+}$. Here, for a given $(\lambda, v)$, the solution satisfying (\ref{design condition}) is generally not unique, leaving degrees of freedom while preserving the constraints. Therefore, to perform a norm-minimizing attack, we first solve $\|\xi^\mathsf{T}X_+\|_2$ and adopt $\xi$ with the smallest norm. This is formulated as the following minimization problem:
\begin{mini!}
    {}          
    {\|\xi^\mathsf{T}X_+\|_2^2}{
    \label{prblem objective xi}}
    {\label{prb xi}}
    \addConstraint{\xi}{\in\mathbb{R}^n \label{eq:constr xi}}
    \addConstraint{\xi^\mathsf{T}X_+v}{=1 \label{eq:constr xi1}}
    \addConstraint{\xi^\mathsf{T}X_+w}{=0, \;\forall w\in\mathcal J^*(\mathcal D). \label{eq:constr xi2}}
\end{mini!}

Fortunately, the solution to this problem can be written in the following closed form: 
\begin{equation*}
\zeta^*:=X_+^\mathsf{T}\xi^*=\frac{\mathrm{proj}_{\mathcal S_+}(v)}{\|\mathrm{proj}_{\mathcal S_+}(v)\|_2^2},\,
\|\zeta^*\|_2^2=\frac{1}{\|\mathrm{proj}_{\mathcal S_+}(v)\|_2^2},
\end{equation*}
where $\mathrm{proj}_{\mathcal S_+}(\cdot)$ is the orthogonal projection onto the subspace $\mathcal S_+:=(X_+^{-1}\Pi_\mathcal O(X_+))^\perp=\mathcal J^*(\mathcal D)^\perp\cap\mathrm{im}\,(X_+^\mathsf{T})\subseteq\mathbb{R}^T$. This is the global minimum solution for the given $v$. Therefore, by designing $\zeta=X_+^\mathsf{T}\xi$ after optimizing $v$, we always obtain $\xi$ with the minimum norm. Substituting this into (\ref{objective function 2}), we have
\begin{equation}
\label{objective function 3}
\mathcal L(\lambda,v)=\frac{\|(\lambda X_--X_++BU_-)v\|_2^2}{\|\mathrm{proj}_{\mathcal S_+}(v)\|_2^2}.
\end{equation}

Based on the above discussion, the minimal-norm attack design problem under $\Delta_{U_-}=0$, $\Delta_{Y_-}=0$, and $\Delta_{X_-}=0$ is given by the following optimization problem:
\begin{mini!}
    {\lambda,v}          
    {\frac{\|(\lambda X_--X_++BU_-)v\|_2^2}{\|\mathrm{proj}_{\mathcal S_+}(v)\|_2^2} \label{prblem objective}}
    {\label{prb}}
    {}
    \addConstraint{v}{\in\mathcal (\mathcal J^*(\mathcal D)^\perp\cap\mathrm{ker}\,(CX_-))\setminus \{0\}  \label{eq:constr1}}
    \addConstraint{v}{\notin X_+^{-1}\Pi_\mathcal O(X_+) \label{eq:constr2}}
    \addConstraint{\zeta^*}{=X_+^\mathsf{T}\xi^*=\frac{\mathrm{proj}_{\mathcal S_+}(v)}{\|\mathrm{proj}_{\mathcal S_+}(v)\|_2^2} \label{eq:constr3}.}
\end{mini!}

An interesting aspect of this optimization problem is that, despite being an attack derived from geometric conditions, the matrix pencil $\lambda X_--X_++BU_-$ used in the data-driven Hautus test naturally appears in the objective function. That is, the matrix appears in the observability rank condition for the noise free case ($M=I_n$ and $N=0$) based on informativity, as demonstrated in (\hspace{-0.01pt}\cite{Eising}, Corollary 11). From this perspective, the minimum-norm attack selects the mode $(\lambda,x_0)$ closest to the observability boundary ---that is, the mode with the smallest Hautus residual $(\lambda X_--X_++BU_-)v=(\lambda I_n-A_{\mathrm{true}})x_0$--- for the original data and pushes it into the unobservable region with minimal displacement. 

To demonstrate this, for a given output matrix $C$, we define the set of all matrices $A$ for which the pair $(A, C)$ becomes unobservable as follows:
\begin{equation*}
\mathrm{UNOBS}:=\{A\in\mathbb{R}^{n\times n}|\;(A,C)\;\text{is unobservable}\}.
\end{equation*}
For a given matrix $A$, a metric between the pair $(A, C)$ and $\text{UNOBS}$ is defined as
\begin{equation*}
d_\text{UNOBS}(A):=\underset{\lambda\in\mathbb C}{\mathrm{inf}}\,\sigma_\mathrm{min}\bigg(\begin{bmatrix}
    \lambda I_n -A\\C
\end{bmatrix}\bigg),
\end{equation*}
where $\sigma_\mathrm{min}(\cdot)$ denotes the smallest singular value of the matrix. This is the observability version of the metric in \cite{REising} for controllability. Here, since the true model $A_\text{true}$ is currently unknown, we consider extending this metric using the affine model set (\ref{affine set}) described by the data $\mathcal D$ as follows:
\begin{equation*}
\label{data diven metric}
d_\text{UNOBS}(\Sigma (\mathcal D)):=\underset{A\in \Sigma (\mathcal D)}{\mathrm{inf}}\,d_\text{UNOBS}(A).
\end{equation*}
This gives the distance from the models consistent with the data $\mathcal D$ to the set of unobservable systems. The next theorem is our second main result of this paper.
\begin{thm}
\label{thm 2}
Suppose that there is no noise in the original data $\mathcal D$, i.e., we set $M=I_n$ and $N=0$. Then, under the linear transformations generated by Algorithm \ref{alg:attack}, the following inequality holds:
\begin{equation*}
\underset{\lambda\in\mathbb C,\,v\in\mathcal F}{\mathrm{inf}}\,\|\Delta\mathcal D\|_\mathrm{F}\geq d_\mathrm{UNOBS}(\Sigma (\mathcal D))\,\sigma_\mathrm{min}(X_-),
\end{equation*}
where $\Delta \mathcal D$ is defined in (\ref{data perturbation}) and $\mathcal F$ is the set of all feasible $v$ for (\ref{eq:constr1}) and (\ref{eq:constr2}).
\end{thm}

We would like to highlight that this theorem states that the minimum data perturbation depends on the unobservability of the system described by the original data. That is, the closer the affine model set is to being unobservable, the smaller the data perturbation can be. This is an intuitively reasonable characterization of the class of attacks considered in this paper. Moreover, this result suggests that anomaly detection based on attack norms may not be effective in certain cases. If the original data are not informative for observability, then there exists at least one unobservable model described by the data, and thus $d_{\mathrm{UNOBS}}(\Sigma(\mathcal D))=0$. That is, there is a possibility that an attack can be constructed with perturbations near zero.
\begin{rem}
This paper limited the analysis to attacks only against $X_+$, but in practice, detection may occur due to time series inconsistencies in both $X_-$ and $X_+$. To counteract this, stealthiness can be enhanced by methods such as tampering only with $U_-$ or $Y_-$, which is left for future research.
\end{rem}

\begin{figure}[t]
\centering
\resizebox{\columnwidth}{!}{%
\begin{tikzpicture}[
    >=Stealth,
    thick,
    state/.style={circle, draw, minimum size=8mm, inner sep=0pt},
    lab/.style={font=\scriptsize, fill=white, inner sep=1.0pt},
    inlab/.style={font=\scriptsize, fill=white, inner sep=1.0pt},
    outlab/.style={font=\scriptsize, fill=white, inner sep=1.0pt}
]

\node[state] (x1) at (0,0)   {$x_1$};
\node[state] (x2) [right=13mm of x1] {$x_2$};
\node[state] (x3) [right=13mm of x2] {$x_3$};
\node[state] (x4) [right=13mm of x3] {$x_4$};
\node[state] (x5) [right=13mm of x4] {$x_5$};

\path[->] (x1) edge[loop above] node[lab] {$0.8$} ();
\path[->] (x2) edge[loop above] node[lab] {$0.7$} ();
\path[->] (x3) edge[loop above] node[lab] {$0.6$} ();
\path[->] (x4) edge[loop above] node[lab] {$0.5$} ();
\path[->] (x5) edge[loop above] node[lab] {$0.4$} ();

\draw[->] (x2) -- node[lab, above] {$0.1$} (x1);
\draw[->] (x3) -- node[lab, above] {$0.1$} (x2);
\draw[->] (x4) -- node[lab, above] {$0.02$} (x3);
\draw[->] (x5) -- node[lab, above] {$0.05$} (x4);

\draw[->] ($(x1.225)+(-12mm,-8mm)$) -- node[inlab, left, yshift=1.0mm] {$u$} (x1.225);
\draw[->] ($(x4.225)+(-12mm,-8mm)$) -- node[inlab, left, yshift=1.0mm] {$u$} (x4.225);

\draw[->] (x1.south) -- node[outlab, right] {$y_1=x_1$} ($(x1.south)+(0,-11mm)$);
\draw[->] (x2.south) -- node[outlab, right] {$y_2=x_2$} ($(x2.south)+(0,-11mm)$);

\end{tikzpicture}%
}
\caption{Linear dynamical network model (5-states)}
\label{fig:linear dynamical network model}
\end{figure}
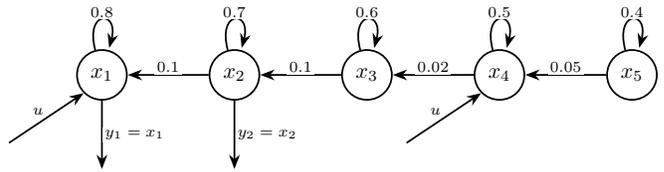

\section{Numerical example}
In this section, we consider the linear dynamical network model shown in Fig.~\ref{fig:linear dynamical network model}, where the state components $x_1,\dots,x_5$ represent network nodes and the directed couplings are encoded in the state matrix $A$. In particular, the network topology is a line where information propagates along a single directed path toward the measured nodes with $y_1=x_1$ and $y_2=x_2$. The line topology provides high interpretability for nodes critical to breaking observability.

The dynamics of this system is expressed as follows:
\begin{equation*}
\begin{split}
x_{k+1}&=\begin{bmatrix}
    0.8 & 0.1 & 0 & 0 & 0\\
     0 & 0.7 & 0.1 & 0 & 0\\
     0 & 0 & 0.6 & 0.02 & 0\\
     0 & 0 & 0 & 0.5 & 0.05\\
     0 & 0 & 0 & 0 & 0.4
\end{bmatrix}x_k+\begin{bmatrix}
    1\\0\\0\\1\\0
\end{bmatrix}u_k,\\y_k&=\begin{bmatrix}
    1 & 0 & 0 & 0 & 0\\
    0 & 1 & 0 & 0 & 0
\end{bmatrix}x_k.
\end{split}
\end{equation*}
This system is observable. For data collection, we set the time duration as $T=100$ and used random initial states without any noise. We computed the minimum-norm attack transformation $\Phi_a^*$ for this data based on (\ref{prb}) and obtained
\begin{equation*}
\begin{split}
\Phi_+^X&=
\begin{bmatrix}
1 & 0 & 0 & 0 & 0 \\
0 & 0.9167 & -0.1763 & 0 & 0 \\
0 & -0.0155 & 0.9672 & 0 & 0 \\
0 & -0.0039 & -0.0082 & 1 & 0 \\
0 & -0.0019 & -0.0040 & 0 & 1 \\
\end{bmatrix},
\end{split}
\end{equation*}
where elements with magnitudes less than $10^{-12}$ were set to zero. In this case, the minimum norm is $\|\Delta_{X_+}\|_{\mathrm{F}}=\|\tilde{X}_+-X_+\|_{\mathrm{F}}=0.0758$, and its relative error is $\epsilon_+^X:=\|\Delta_{X_+}\|_{\mathrm{F}}/\|X_+\|_{\mathrm{F}}=3.7\times10^{-3}$. Also, the unobservable eigenpair embedded in $\mathcal V^*(\tilde{\mathcal D})$ was taken as $\tilde{\lambda}^*=0.5014$ and $\tilde{x}_0^*=\begin{bmatrix}
    0 & 0 & -0.0194 & 0.0776 & 0.0004
\end{bmatrix}^\mathsf{T}$.

Now, we quantify the nodes where the data manipulation is concentrated. To do so, we measured the row-wise modification magnitude $\delta_i := \|e_i^\top \Delta_{X_+}\|_2$, $i=1,\dots,5$, where $e_i$ is the vector of the $i$-th canonical basis. Then, we can write $\|\Delta_{X_+}\|_\mathrm{F}^2=\sum_{i=1}^5 \delta_i^2$, and the corresponding attack contribution ratio of row $i$ is expressed as
\begin{equation*}
\rho_i := \frac{\delta_i^2}{\|\Delta_{X_+}\|_\mathrm{F}^2}\in[0,1],\qquad \sum_{i=1}^5\rho_i=1.
\end{equation*}

Based on the simulated data of 100 steps, we found that the manipulation is highly localized as the attack contribution ratios of nodes 2 and 3 are especially high: $\rho_2\approx 0.9640$ and $\rho_3\approx 0.0334$. The remaining rows contribute negligibly (and the first row is essentially zero up to numerical precision). This means that most of the attack energy is spent on modifying the update equations of $x_2$ and $x_3$.

This localization admits a clear network-structural explanation in terms of hop distance to the measured nodes 1 and 2. Let the hop distance $d(i)$ as the length of the shortest directed path from node $i$ to either measured node. In a network with a line topology, nodes with shorter hop distances influence the measured nodes through fewer links. Consequently, to induce a comparable degradation of observability as seen from the measurements, perturbations applied closer (in hops) to the measured nodes are typically more effective, and a minimum-energy manipulation tends to concentrate on low-hop nodes. Therefore, it is important to focus on defending the vulnerable nodes identified based on such network topologies.

\section{Conclusion}
In this paper, we have studied malicious data transformations against informativity-based analysis in data-driven control, focusing on strong observability. We have shown that an invertible linear transformation of finite trajectory data can destroy strong-observability informativity by embedding a nonzero state into the data-driven weakly unobservable subspace. We have also provided constructive attack designs and a minimum-norm formulation to quantify the smallest distortion required to invalidate informativity. In future work, we will extend our framework to other properties such as controllability and stabilizability. In addition, attack detection and data recovery methods based on the transformation of invariant subspaces are important challenges.


\begin{thebibliography}{99}
\bibitem{DePersis}
C.~De Persis and P.~Tesi, ``Formulas for data-driven control: Stabilization,
optimality, and robustness,'' \textit{IEEE Trans. Autom. Control}, vol.~65, no.~3,
pp.~909--924, 2020.

\bibitem{Informativity}
H.~J. van Waarde, M.~K. Camlibel, and H.~L. Trentelman, \textit{Data-Based Linear Systems and Control Theory}, Kindle Direct Publishing, 2025.

\bibitem{Mishra}
V.~K. Mishra, H.~J. van Waarde, and N.~Bajcinca, ``Data-driven criteria
for detectability and observer design for lti systems,'' in \textit{Proc. 61st IEEE Conf.
Decis. Control}, pp.~4846--4852, 2022.

\bibitem{DCS} J.~Eising and J.~Cortés, ``Informativity for centralized design of distributed controllers for networked systems,'' in \textit{Proc. Eur. Control Conf.}, pp.~681--686, 2022.

\bibitem{MAS} J.~Jiao, H.~J. van Waarde, H.~L. Trentelman, M.~K. Camlibel, and S.~Hirche, ``Data-driven output synchronization of heterogeneous leader-follower
 multi-agent systems,'' in \textit{Proc. IEEE Conf. Decis.
Control}, pp.~466--471, 2021.

\bibitem{Necsys25} I.~Takaki, A.~Cetinkaya, and H.~Ishii, ``Trade-off in quantization between data-driven design and control inputs,'' in \textit{Proc. 10th IFAC Conference on Networked Systems}, pp.~103--108, 2025.

\bibitem{Russo} A.~Russo, ``Analysis and detectability of offline data poisoning attacks on linear dynamical systems,'' in \textit{Proc. Learning for Dynamics and Control Conf.}, pp.~1086--1098, PMLR, 2023.

\bibitem{Sasahara} H.~Sasahara, ``Adversarial destabilization attacks to direct data-driven control,'' \textit{arXiv preprint arXiv:2507.14863}, 2025.

\bibitem{Russo Proutiere} A.~Russo and A.~Proutiere, ``Poisoning attacks against data-driven control methods,'' in \textit{Proc. Amer. Control Conf.}, pp.~3234--3241, 2021.

\bibitem{Anand}  S.~C. Anand, M.~S. Chong, and A.~M. Teixeira, ``Data-driven
identification of attack-free sensors in networked control systems,'' \textit{arXiv preprint arXiv:2312.04845}, 2023.

\bibitem{Yan} J.~Yan, I.~Markovsky, and J.~Lygeros, ``Secure data reconstruction: A
direct data-driven approach,'' \textit{IEEE Trans. Autom. Control}, vol.~70, no.~12, pp.~8361--8367, 2025.

\bibitem{Shinohara} T.~Shinohara, K.~H. Johansson, and H.~Sandberg, ``Data-driven
resilience assessment against sparse sensor attacks,'' \textit{arXiv preprint arXiv:2509.25064}, 2025.

\bibitem{TanakaKanekoKyb}
Y.~Tanaka and O.~Kaneko,
``Data transformation technique in the data informativity approach via algebraic sequences,''
\textit{Kybernetika}, vol.~60, no.~2, pp.~228--243, 2024.

\bibitem{Eising} J.~Eising and H.~L. Trentelman,``Informativity of noisy data for structural properties of linear systems,'' \textit{Syst. Control Lett.}, vol.~158, p.~105058, 2021.

\bibitem{EisingECC}
J.~Eising, H.~L. Trentelman, and M.~K. Camlibel,
``Data informativity for observability: An invariance-based approach,''
in \textit{Proc. European Control Conference (ECC)}, pp.~1057--1059, 2020.

\bibitem{Trentelman} H.~L. Trentelman, A.~A. Stoorvogel, and M.~Hautus, \textit{Control Theory for
Linear Systems}, London, U.K.: Springer-Verlag, 2001.

\bibitem{Willems} J.~C. Willems, P.~Rapisarda, I.~Markovsky, and B.~L. M. De Moor, ``A
note on persistency of excitation,'' \textit{Syst. Control Lett.}, vol.~54, no.~4, pp.~325--329, 2005.

\bibitem{Multiple datasets} H.~J. van Waarde, C.~De Persis, M.~K. Camlibel, and P.~Tesi, ``Willems' fundamental lemma for state-space systems and its extension to multiple datasets,'' \textit{IEEE Control Syst. Lett.}, vol.~4, no.~3, pp.~602--607, 2020.

\bibitem{DePersis PE}
C.~De Persis and P.~Tesi, ``On persistency of excitation and formulas
for data-driven control,'' in \textit{Proc. 62nd IEEE Conf.
Decis. Control}, pp.~873--878, 2019.

\bibitem{REising} R.~Eising, ``Between controllable and uncontrollable,'' \textit{Syst. Control Lett.}, vol.~4, no.~5, pp.~263–264, 1984.


\end{thebibliography}
\end{document}